\documentclass{theoretics}

\ThCSauthor[1]{Shahar Dobzinski}{shahar.dobzinski@weizmann.ac.il}[0000-0003-1935-5808]
\ThCSauthor[2]{Sigal Oren}{orensi@bgu.ac.il}[0000-0002-4271-7291]
\ThCSauthor[3]{Jan Vondr\'ak}{jvondrak@stanford.edu}[0009-0001-6021-679X]
\ThCSaffil[1]{Weizmann Institute of Science, Israel}
\ThCSaffil[2]{Ben Gurion University of the Negev, Israel}
\ThCSaffil[3]{Stanford University, USA}

\title{Fairness and Incentive Compatibility via Percentage Fees}

\ThCSshortnames{S.~Dobzinski, S.~Oren and J.~Vondr\'ak}
\ThCSshorttitle{Fairness and Incentive Compatibility via Percentage Fees}
\ThCSyear{2026}
\ThCSarticlenum{12}
\ThCSreceived{Jul 12, 2024}
\ThCSrevised{Oct 17, 2025}
\ThCSaccepted{Jan 5, 2026}
\ThCSpublished{Jul 22, 2026}
\ThCSdoicreatedtrue
\ThCSkeywords{Mechanism Design, Nash Social Welfare, Fairness}

\usepackage[all]{xy}
\usepackage{tikz}
\usetikzlibrary{arrows.meta,chains,decorations.pathreplacing}
\usepackage[normalem]{ulem}
\declaretheorem[style=italicized,numbered=no,name={Theorem\kern-.09em}]{theorem*}

\newcommand{\eps}{\varepsilon}

\newcommand{\stam}[1]{}

\newcommand{\R}{\mathbb{R}}

\newcommand{\poly}{\hbox{poly}}

\newcommand{\FF}{\mathbb F}

\newcommand{\E}{\mbox{\bf E}}

\newcommand{\cA}{{\mathcal A}}

\newcommand{\cD}{{\mathcal D}}
\newcommand{\cM}{{\mathcal M}}

\newcommand{\cR}{{\mathcal R}}

\newcommand{\cV}{{\mathcal V}}

\begin{document}

\ThCSthanks{We thank Rahul Deb, Vasilis Gkatzelis, Nima Haghpanah, and Rachel Kranton for helpful discussions and various pointers to the literature.\\
S.~Dobzinski is supported by ISF grant 2185/19. S.~Dobzinski and J.~Vondrák are supported by BSF-NSF grant (BSF number: 2021655, NSF number: 2127781). S.~Oren is supported by BSF grant 2018206 and ISF grant 2167/19.\\
A preliminary version of this article appeared in EC 2023~{\cite{DBLP:conf/sigecom/DobzinskiOV23}}.}

\maketitle

\begin{abstract}
We study incentive-compatible mechanisms that maximize the Nash Social Welfare. Since traditional incentive-compatible mechanisms cannot maximize the Nash Social Welfare even approximately, we propose changing the traditional model. Inspired by a widely used charging method (e.g., royalties, a lawyer that charges some percentage of possible future compensation), we suggest charging the players some percentage of their value of the outcome. We call this model the \emph{percentage fee} model.

We show that there is a mechanism that maximizes exactly the Nash Social Welfare in every setting with non-negative valuations. Moreover, we prove an analog of Roberts theorem that essentially says that if the valuations are non-negative, then the only implementable social choice functions are those that maximize weighted variants of the Nash Social Welfare. We develop polynomial time incentive compatible approximation algorithms for the Nash Social Welfare with subadditive valuations and prove some hardness results. 
\end{abstract}

\section{Introduction}

The field of Mechanism Design aims to develop and analyze algorithms for strategic players. In a typical scenario, we have a planner interested in implementing some social goal. The challenge is to design incentive-compatible mechanisms that achieve this social goal, despite the players behaving in a strategic way that might be misaligned with the desired social goal.

In this paper, we are interested in mechanisms that are dominant strategy incentive compatible when the values of the players are private information. As for the social goals -- many different ones are studied in the literature. However, largely speaking, it is fair to say that (with some very notable exceptions) the two most common and well-studied objectives are social welfare maximization \cite{Vic61,Cla71,Gro73} and revenue maximization \cite{Mye81}. These two social goals are very different: social welfare maximization (i.e., outputting an outcome that maximizes the sum of values of the players) is an objective that is defined for each instance, independently of the assumptions on the strategic behavior of the players. In contrast, revenue maximization is relative -- the quality of a mechanism is measured with respect to its closeness to the revenue of some ``ideal'' mechanism.

Taking a bird's-eye view, most will agree that welfare maximization is by far the most successful ``benchmark-free'' objective in the mechanism design literature and that good mechanisms for other ``benchmark-free'' objectives usually exist only for quite restricted settings (e.g., makespan minimization in the so-called ``single parameter'' settings). This grim situation is, of course, not due to the incapability of mechanism designers: it is possible to prove that incentive-compatible mechanisms can only achieve few objectives. In fact, Roberts theorem \cite{Rob79} tells us that social welfare maximization is unique in that in some settings, the set of implementable objective functions includes only slight variants of social welfare maximization.

Indeed, we have a good understanding of which social goals can and cannot be achieved by incentive-compatible mechanisms in many settings of interest. Obviously, being able to mathematically prove that incentive-compatible mechanisms are not powerful enough in some settings is of significant academic interest. Still, it is disappointing news from a practical perspective. Fairness is a case in point. In recent years we have seen a surge of interest in fairness. Notions such as EF1 (\cite{budish2011combinatorial} envy freeness up to one good), EFX (\cite{CKMPSW19} envy freeness up to any good), and the Nash Social Welfare (the product of the values of the players) have been extensively studied. Much of the work focused on existence and ``algorithmic'' issues: does a ``fair'' outcome exist in every instance? Can we find this outcome efficiently? Yet, even for the most extensively-studied fairness notions, we have no good mechanisms that implement them if the agents are strategic, except perhaps for relatively simple settings (e.g., \cite{BEF21,CFFKO11} ). {The work of Cole et al. \cite{CGG13} stands out as a notable exception and, in a sense, serves as the starting point of this study, as discussed later.} For example, in \cite{CFFKO11} a particular class which slightly generalizes additive valuations is shown to admit a mechanism that produces a truthful and envy-free solution, but it is shown that no such mechanism exists for slight generalizations of this class\footnote{An interesting singular exception is running the VCG mechanism when all players are unit demand: VCG outputs not only the welfare maximizing solution but also one that is envy-free \cite{L83}.}.

This paper attempts to bridge the chasm between incentives and fairness. We aim to design incentive-compatible mechanisms for one of the most prominent fairness promoting objectives, the Nash Social Welfare. Unfortunately, it is not hard to see that dominant strategy mechanisms cannot always output the allocation with the highest Nash Social Welfare, nor can they always output an allocation that provides a reasonable approximation to it\footnote{To see this, consider a dominant strategy algorithm for two players, Alice and Bob, and two items, $a$ and $b$. The valuation of both players is additive. Consider an instance where Alice's valuation is $1$ for item $a$ and $x$ for item $b$, and Bob's value is $x$ for $a$ and $x^3$ for item $b$. Suppose that $x>>1$. The only allocation that reasonably approximates the Nash social welfare is to give $a$ to Alice and $b$ to Bob. Now consider an instance where Alice's values are the same, but Bob's values are $t$ for item $a$ and $t+x^3$ for item $b$, $t>>x$. By weak monotonicity, Bob is always allocated item $b$, and possibly also item $a$. However, Bob cannot be allocated both items because then the Nash social welfare of the outcome will be $0$. Thus, the mechanism must output the allocation that gives Alice item $a$ and Bob item $b$. Note that the Nash social welfare of the allocation that gives Alice item $b$ and Bob item $a$ is bigger by a factor of $x$. Thus no dominant strategy mechanism can obtain a reasonable approximation to the Nash social welfare.}. 

In this paper we suggest to reconsider the traditional payment model. The taxation principle tells us that in the traditional model of mechanism design, each player is (essentially) facing a menu that sets a price for each possible alternative. This corresponds to a common real-life fee type known as a ``fixed fee''. However, another common fee type is the ``percentage fee''. Percentage fees might be used, e.g., by a lawyer who may charge the client a portion of the future compensation or
in a lease agreement of a retail store that commits to paying a percentage of its sales as rent. Royalties are another example for percentage fees. 

This paper's take-home message is that percentage fees are an efficient way of constructing fair incentive-compatible mechanisms. We view our results as a way of escaping the dead end that the traditional mechanism design model leads us to as far as implementing fairness notions is concerned. 

\subsubsection*{Our Model}

This work primarily studies fair dominant-strategy mechanisms in a combinatorial auction setup. However, the model is defined for the most general mechanism design setting, and some of our results also apply to this general model.

\paragraph{The setup.} In the most general setup, we have a set $N$ of players and a set $\cA$ of alternatives. Each player $i$ has a valuation function $v_i:\mathcal A \rightarrow \mathbb R$. The set of all possible valuation functions of player $i$ is denoted $\mathcal V_i$. A (direct) mechanism is composed of a social choice function $f:\mathcal V_1 \cdots \mathcal V_n\rightarrow \mathcal A$ and a payment function $p:\mathcal V_1 \cdots \mathcal V_n\rightarrow \mathbb R^n$. 

\paragraph{Incentive Compatibility.} Much of the mechanism design literature assumes that the profit of player $i$ in the instance $\overrightarrow v=(v_1,\ldots, v_n)$ is $v_i(f(\overrightarrow v))-p_i(\overrightarrow v)$ and looks for dominant strategy mechanism given this definition of profit. I.e., in a dominant strategy mechanism, for each player $i$, valuations $v_i,v'_i\in \mathcal V_i$ and valuations of the other players $v_{-i}$ it holds that
\begin{equation}\label{eq-traditional-ic}
v_i(f(v_i,v_{-i}))-p_i(v_i,v_{-i}) \geq v_i(f(v'_i,v_{-i}))-p_i(v'_i,v_{-i})
\end{equation}
We refer to this model as the \emph{traditional} model and refer to mechanisms where each player has a dominant strategy as in Equation (\ref{eq-traditional-ic}) as incentive compatible in the traditional model.

In contrast, in the \emph{percentage fee} model each player $i$ is charged a fraction of his value: $p_i(\overrightarrow v)\cdot v_i(f(\overrightarrow v))$. Thus, the profit of player $i$ in the instance $\overrightarrow v=(v_1,\ldots, v_n)$ is $v_i(f(\overrightarrow v))\cdot (1-p_i(\overrightarrow v))$. We are also interested in dominant strategy mechanisms in this model, but note that the definition of dominant strategy now considers the new profit model. That is, for each player $i$, valuations $v_i,v'_i\in \mathcal V_i$ and valuations of the other players $v_{-i}$ it holds that
\begin{equation}\label{eq-percentage-ic}
(1-p_i(v_i,v_{-i}))\cdot v_i(f(v_i,v_{-i})) \geq (1-p_i(v'_i,v_{-i}))\cdot v_i(f(v'_i,v_{-i}))
\end{equation}
We refer to mechanisms where each player has a dominant strategy as in Equation (\ref{eq-percentage-ic}) as incentive compatible in the percentage fee model.

We say that $f$ is \emph{implementable} (in the percentage fee model or in the traditional model) if there exists a payment function $p$ such that $(f,p)$ is an incentive compatible mechanism (in the relevant model).

\paragraph{Individual rationality.} We will mainly be interested in individually rational mechanisms with no positive transfers. In the percentage-fee model, this means that for every $i$, $p_i(\cdot)$ takes values in $[0,1)$ only\footnote{We do not allow $p_i(\cdot)=1$ to rule out trivial mechanisms that make little economic sense, like setting $p_i(\vec v)=1$ for every player $i$ in instance $\vec{v}$. In this mechanism the profit of all players is always $0$ (since if alternative $a$ is chosen the mechanism charges each player $i$ its full value $v_i(a)$), so pairing this payment function with any allocation function will give an incentive compatible mechanism.}. We stress that the players still have quasilinear utilities as before. I.e., they want to maximize their value for the selected alternative minus the payment. 

\paragraph{Nash social welfare.} Our main focus in this paper is maximizing the Nash Social Welfare in combinatorial auctions by dominant strategy mechanisms in the percentage fee model. In a combinatorial auction, we have a set $N$ of players ($|N|=n)$ and a set $M$ of heterogeneous items ($|M|=m$). Each player $i$ has a valuation function $v_i:2^M\rightarrow \mathbb R$ that gives a value for each possible subset of the items. We assume that each valuation function $v$ is non-decreasing and $v(\emptyset)=0$. One common goal is to maximize the social welfare of the output allocation $(S_1,\ldots, S_n)$: $\Sigma_iv_i(S_i)$. In this paper, our primary focus is to maximize the Nash Social Welfare (NSW), where the NSW of an allocation $S_1,\ldots, S_n$ is the {\em geometric mean} of the valuations $(\prod_{i=1}^{n} v_i(S_i))^{1/n}$.

\paragraph{Valuation classes.} We consider several standard classes of valuations in this paper: 
\begin{itemize}
\item A valuation $v$ is \emph{additive} if for every bundle $S$, $v(S)=\Sigma_{j\in S}v(\{j\})$.
\item $v$ is $\emph{XOS}$ if there exists additive functions $a_1,\ldots, a_t$ such that for every bundle $S$, it holds that $v(S)=\max_{1\leq j\leq t}a_j(S)$.
\item $v$ is \emph{subadditive} if for every $S,T$, $v(S)+v(T) \geq v(S\cup T)$. 
\end{itemize}

It is well known and easy to see that every additive valuation is also XOS, and that every XOS valuation is also subadditive.

\subsubsection*{Applicability of Mechanisms with Percentage Fees}

Mechanisms with percentage fees are applicable whenever the auctioneer can learn how much a player values the \emph{selected} outcome. We stress that 
the auctioneer does not necessarily learn

the values of other unrealized outcomes. For example, one can auction a license for the right to use some asset in exchange for a fraction of the future revenue, where the latter can be verified, e.g., by official tax returns. Our mechanisms are particularly attractive when the auction goal is to maximize fairness. One extreme example would be land reforms, where (typically) agricultural land is redistributed by the government to maximize both efficiency and equality. Land reforms also take a less radical form: In some countries, e.g., Israel, it is not uncommon to repartition land rights and move land from a ``strong'' municipality to a neighboring, economically weaker one. In the US, mechanisms that are used to redistribute resources to promote equality include gaming rights for native American tribes. Similarly, our mechanisms are applicable when resources are allocated internally within a large organization or corporation, where the management can evaluate the value of the allocated resources for the winning units.

Our work is related to the work on contingent payments (see \cite{Hansen85, DM2014} and the survey \cite{S2013}, among others). A typical scenario in this research direction includes multiple firms competing on acquiring a target firm, where the seller of the target firm receives, for example, a combination of cash and some percentage of the merged firm.

Similarly, \emph{sharecropping} is a common type of legal agreement that can be seen as a mechanism with percentage fees. In sharecropping, the landowner allows the use of the land in return to a share of the crop. Thus, assuming that the tenant's utility is linear in the amount of produced crop, the landowner can charge a percentage fee even without knowing the precise value of the tenant or the crop. Sharecropping was extensively studied, in particular in the economic literature. See, e.g., an influential paper by Stiglitz \cite{stiglitz74}.

The mechanisms we design aim to (approximately) maximize the product of the values (the Nash social welfare), rather than the product of the profits. While mechanisms that maximize the product of the profits might sometimes be considered ``fairer,'' focusing on mechanisms for Nash social welfare aligns more closely with the previous work in mechanism design which studies implementability of allocation functions in different settings (as mentioned above). Furthermore, it is quite common to evaluate the fairness of an algorithm based on its final allocational output without considering the players' monetary holdings. Consequently, changes in the players' monetary resources should not typically affect the perceived fairness of the allocation of goods.

\subsubsection*{Our Results I: Incentive Compatible Mechanisms that Maximize the NSW}

The Nash Social Welfare has been heavily studied recently in Algorithmic Game Theory. Its game theoretic properties have been analyzed (e.g., \cite{CKMPSW19}) and the possible approximation ratios achievable in different settings have been studied (e.g., \cite{CV15,AOSS17,BKKV18,LV22}). Unfortunately, as discussed earlier, no dominant strategy mechanisms can maximize the NSW in the traditional model. In the percentage fee model, we observe the following in Subsection \ref{subsec-exact-nash}:

\begin{theorem*}
In the percentage fee model, the social choice function that selects an alternative that maximizes the Nash Social Welfare is implementable as long as all valuations are positive or all valuations are non-negative and there is a ``null'' alternative with value $0$ for all players. 
\end{theorem*}

{This result can be seen as a generalization of an earlier mechanism of Cole et al. \cite{CGG13}, see below for a discussion.}

We have that just as the VCG mechanism is always applicable in the traditional model, maximizing the NSW is  possible in the percentage fee model. Recall that in the traditional model, the VCG mechanism is unique as Roberts theorem tells us that if the valuations are unrestricted, then the only implementable social choice functions are weighted variations of maximizing the social welfare. In Subsection \ref{subsec-analog} we prove that in the percentage fee model, only Nash Social Welfare maximization variants are implementable.

\begin{theorem*}
Let $\mathcal M$ be an incentive compatible mechanism in the percentage fee model when the valuations are positive but otherwise unrestricted. Suppose that the size of the image of the allocation function of $\cM$ is at least $3$. Then, there exist constants $\alpha_a$ (for each $a\in \mathcal A$) and $\beta_i$ (for each player $i$) such that in every instance $(v_1,\ldots, v_n)$ the allocation function of $\cM$ outputs an alternative that belongs to $\arg\max_{a\in \mathcal A}\alpha_a\cdot \prod_iv_i(a)^{\beta_i}$. 
\end{theorem*}

 We do not prove this theorem directly but instead present a reduction from Roberts theorem. The reduction relies on a simple yet powerful observation. 
Given a positive valuation $v$, let the valuation ${\tt log} v:2^M \to \R$ be defined as ${\tt log} v(S) = \log v(S)$. Now, let $\cV$ be a class of non-negative valuations and let ${\tt log} \cV$ denote the class of valuations: $ {\tt log} \cV = \{ {\tt log} v: v \in \cV \}$. We show a one-to-one and onto correspondence between mechanisms in the traditional model when each valuation is in ${\tt log} \cV$ and mechanisms in the percentage fee model when each valuation is in $\cV$. Note that if $\cV$ is the set of unrestricted positive valuations, then ${\tt log} \cV$ is the set of unrestricted valuations. Thus, Roberts theorem applies to the set of implementable social choice functions in the traditional model when each valuation is in ${\tt log} \cV$. We use this to characterize the set of implementable social choice functions when the valuations are in $\cV$ in the percentage fee model. 

The inspiration for our work comes from Cole et al. \cite{CGG13}, which is the first paper to observe a connection between implementability of the Nash social welfare and logV. They construct a mechanism that maximizes the Nash Social Welfare. Their mechanism can be seen as an implementation of the VCG mechanism in $ {\tt log} \cV$. In their setup the valuations are homogeneous, so this implementation does not require money, similarly to the sharecropping setting discussed above. {Note that although the mechanism of Cole et al. \cite{CGG13} works in the more restrictive setting of homogeneous valuations, it does not require monetary transfers at all. Thus, our mechanism can be seen as pushing their mechanism to its limit, by noticing that the mechanism of Cole et al. \cite{CGG13} can be extended to more general settings, if monetary transfers of a particular form are allowed}.

We also study the set of implementable functions in the single parameter setting. We focus on binary single parameter domains where  for each player $i$, the set of alternatives $\mathcal A$ is divided into ``winning'' alternatives $\mathcal W_i$ and ``losing'' alternatives $\mathcal L_i$. For every player $i$ and valuation $v_i\in \mathcal V_i$ there is a value $h_{v_i}$ such that $v_i(a)=h_{v_i}$, for all $a\in \mathcal W_i$. For every player $i$, there exists a value $l_i>0$ such that for every valuation $v_i\in \mathcal V_i$ and alternative $a\in \mathcal L_i$, $v_i(a)=l_i$. Unlike implementability in rich domains, in single-parameter settings we get that the set of implementable allocation functions is identical in the percentage fee model and in the traditional model (Subsection \ref{subsec-single-parameter}):

\begin{theorem*}
Let $f$ be an allocation function when the domains of all players are binary single parameter. Then, $f$ is implementable in the percentage fee model if and only if $f$ is monotone\footnote{Recall that $f$ is monotone for player $i$ if for each $v_i, v_{-i}$ for which $i$ wins in the instance $(v_i,v_{-i})$, $i$ also wins in the instance $(v'_i,v_{-i})$ when $v'_i>v_i$.} for each player $i$.    
\end{theorem*}

The characterization results that we provide in this paper are similar in spirit to a paper by Deb and Mishra \cite{DM2014} on implementability in domains with contingent payments. However, the set of allowed payments in \cite{DM2014} is different than ours as well as some of their assumptions (e.g., their ``binary independence'' assumption). Thus, the characterization results of neither paper imply the other. Furthermore, the techniques used to obtain the characterization results vary significantly.

\subsubsection*{Our Results II: Computationally Efficient Approximation Mechanisms}

We have seen that maximizing the Nash Social Welfare is possible with percentage fees. However, maximizing the Nash Social Welfare is NP-hard even in very simple settings, e.g., when the valuations are additive. Thus, much work has focused on developing approximation algorithms for the NSW, e.g., a constant approximation for combinatorial auctions with submodular valuations \cite{LV22,GHLVV22}.

One could hope that the correspondence that is used to obtain the analog of Roberts theorem would enable the ``automatic conversion'' of every computationally efficient approximation mechanism in the traditional model to an approximation mechanism in the percentage fee model with a comparable approximation ratio. Unfortunately, the correspondence does not preserve the approximation ratio. We thus must develop new computationally efficient and incentive-compatible approximation mechanisms for the percentage fee model.

One obstacle in constructing good incentive compatible mechanisms is that simple mechanisms for maximizing the social welfare do not provide any reasonable approximation ratio for the Nash Social Welfare. For example, both the mechanism that allocates the grand bundle to the player that values it the most and the mechanism that randomly allocate the items provide an $n$ approximation to the social welfare, but the first mechanism provides no approximation to the NSW and a random allocation might output an instance with non-zero NSW with exponentially small probability\footnote{Suppose we have $n$ players and $n$ items. Each player is interested only in one unique item for a value of $1$ and the rest for a value of $0$. The only allocation that gives a positive NSW is the one that gives each player his unique item.}. To overcome this obstacle, we ``derandomize'' this random allocation mechanism and get a deterministic mechanism that runs in polynomial time when the number of players $n$ is constant (Subsection \ref{subsec-approx-constant-number}):

\begin{theorem*}
Consider a combinatorial auction with $m$ items and $n$ players with XOS valuations. There is a deterministic mechanism in the percentage fee model that guarantees an approximation ratio of $(1+\epsilon) n$, for any constant $\epsilon>0$. If the valuations are subadditive, then the approximation ratio is $O(n \log m)$. The running time of the mechanism is $O(m^n)$ (i.e., $\poly(m)$ for a constant number of players).
\end{theorem*}

The mechanism enumerates over the support of an $n$-wise independent distribution, where at least one of them provides a good approximation to the NSW. We develop the percentage fee analog of (traditional) maximal-in-range mechanisms to prove incentive compatibility.

Maximal-in-range mechanisms were heavily studied in the traditional model as a way of obtaining computationally efficient dominant strategy mechanisms \cite{DN07a,DN07b,BDFK10,DSS15}. In a maximal-in-range mechanism, there is a fixed set of allocations (independent of the input) and the mechanism finds the welfare-maximizing allocation in the range. If the welfare-maximizing allocation in the range always guarantees a good approximation ratio and can be efficiently found, then applying the VCG mechanism (with respect to the restricted range) ensures incentive compatibility. Similarly, for maximizing the NSW in the percentage fee model, we identify a restricted range in which the best allocation in the range always has a high NSW. The range structure will be simple enough so that the best allocation can be efficiently found. Incentive compatibility of maximal-in-range mechanisms in the percentage fee model is proved similarly to proving the incentive compatibility of maximizing the NSW that was discussed above.

\begin{theorem*}
Consider a combinatorial auction with $m$ items and $n$ players with subadditive valuations. There is an incentive compatible $O(\min(n^2,\frac m n))=O(m^{\frac 2 3})$-approximation mechanism in the percentage fee model. The mechanism makes $\poly(m,n)$ value queries and runs in $\poly(m,n)$ time.
\end{theorem*}

This mechanism can be found in Subsection \ref{subsec-approx-m23}. Note that in this setting, the best known approximation algorithm (which is not incentive compatible) is that of \cite{BKKN21} which provides an approximation ratio of $n^{\frac {53} {54}}$, but this algorithm uses demand queries. Using value queries, there was a known $O(n)$-approximation \cite{GKK20,BBKS20}, and in terms of the number of items the best approximation ratio that one can hope for algorithms for subadditive valuations that use polynomially many value queries, even ignoring incentives, is $O(\sqrt m)$ \cite{DNS05-journal}\footnote{Formally, the result of \cite{DNS05-journal} applies to maximizing the social welfare, not the NSW. However, if in a lower bound proof for the social welfare, the optimal welfare maximizing solution gives all players approximately the same value, then the arithmetic/geometric mean inequality suffices to claim the same bound for NSW. This is the case for this particular lower bound proof, and in fact all lower bounds for the social welfare in all oracle models that we are aware of have this property.}. In Subsection \ref{subsec-approx-m12} we match this approximation ratio with incentive-compatible mechanism with polynomially many value queries but with an exponential running time:

\begin{theorem*}
Consider a combinatorial auction with $m$ items and $n$ players with subadditive valuations. There is an incentive compatible maximal-in-range mechanism in the percentage fee model that guarantees an approximation ratio of $\tilde O(\min(n,\frac m n))=\tilde O(m^{\frac 1 2})$. The mechanism uses $\poly(m,n)$ value queries and $2^n \poly(m)$ running time.
\end{theorem*}

Our last result (Subsection \ref{subsec-hardness}) shows that even for a constant number of players $n$, no polynomial time maximal-in-range mechanism can guarantee an approximation ratio better than $n$. This is tight, considering the maximal-in-range mechanisms discussed above. We also note that without incentive-compatibility, a constant factor of $1-1/e-\epsilon$ can be achieved for NSW with any constant number of players with submodular valuations \cite{GKK20}.

\begin{theorem*}
Let $\cM$ be a maximal-in-range mechanism for $n$ players with valuations from a class that includes additive and valuations. Suppose that $\cM$ guarantees an approximation ratio of $1/n+\eps$, for some constant $\eps>0$. Then, $\cM$ does not run in polynomial time, unless $NP\subseteq P/\poly$.
\end{theorem*}

\section{Implementability in the Percentage Fee Model}

In this section we will see that maximizing the Nash Social Welfare is possible with percentage fees. We will then prove an analog of Roberts theorem by showing that if the valuations are positive but unrestricted, the only set of implementable social choice functions in the percentage fee model are those that maximize weighted versions of the Nash Social Welfare. Finally, we study single-parameter mechanisms in the percentage fee model.

\subsection{Maximizing the Nash Social Welfare}\label{subsec-exact-nash}

\begin{theorem}[Nash Social Welfare Maximization]
\label{thm:VCG}
Consider a domain $\mathcal V$ where all valuations are positive or all valuations are non-negative and there is an alternative $null$ such that $v_i(null)=0$ for every player $i$. Let $f$ be a social choice function defined on $\mathcal V^n$ that selects an alternative that maximizes the Nash Social Welfare. Suppose that when the optimal Nash Social Welfare is $0$, $f$ selects the alternative $null$. Then, $f$ is implementable in the percentage fee model.

\end{theorem}

Note that the theorem holds in particular for combinatorial auctions, where all valuations are non-negative and the allocation that does not allocate any items has value $0$ for all players. In fact, the use of this condition is a technicality and the function that maximizes the Nash Social Welfare can be implemented in any domain if we do allow $p_i(v_1,\ldots, v_n)=1$, but it suffices to allow that only if the Nash Social Welfare in the instance $(v_1,\ldots, v_n)$ is $0$.

\begin{proof}
Let $p_i$ be the following function:

$$
p_i(v_1,\ldots,v_n) = 
\begin{cases}
1 -  \frac{\prod_{j \neq i} v_j(f(v_1,\ldots, v_n))}{\max_{a \in \cA} \prod_{j \neq i} v_j(a)} & \text{if $\max_{a \in \cA} \prod_{j} v_j(a)> 0$};\\
0 & \text{if $\max_{a \in \cA} \prod_{j} v_j(a)= 0$}.
\end{cases}$$
Observe first that the mechanism $(f,p)$ is well defined, as for every player $i$ and $v_1,\ldots, v_n$ we have that $0\leq p_i(v_1,\ldots, v_n)<1$.

To see that this payment function is incentive compatible, start by fixing $v_{-i}$. We first handle the case in which $\max_{a \in \cA} \prod_{j} v_j(a)> 0$. We have that $(1-p_i(v_i, v_{-i}))\cdot v_i(f(v_i, v_{-i}))=\frac {\left(\prod_{j\neq i}v_j(o) \right)\cdot v_i(o)} {\max_{a \in \cA} \prod_{j \neq i} v_j(a)}$, for  $o=f(v_i,v_{-i})\in \arg\max_a \prod_{j}v_j(a)$. For any other $v'_i$, we have that $(1-p_i(v'_i, v_{-i}))\cdot v_i(f(v'_i, v_{-i}))=\frac {\left(\prod_{j\neq i}v_j(w) \right)\cdot v_i(w)} {\max_{a \in \cA} \prod_{j \neq i} v_j(a)}$, for $w=f(v'_i,v_{-i})$. Incentive compatibility follows by $\left(\prod_{j\neq i}v_j(o) \right)\cdot v_i(o) \geq \left(\prod_{j\neq i}v_j(w)\right)\cdot v_i(w)$ since $o$ maximizes the NSW. 

Now, consider the case in which $\max_{a \in \cA} \prod_{j} v_j(a)= 0$. Recall that in this case the mechanism outputs the null alternative which all players value at $0$. Observe that for every $v'_i$ it is either the case that $f(v'_i,v_{-i})=null$ or that $f(v'_i,v_{-i})=w\neq null$ and $v_i(w)=0$. In either case we have that $v_i(f(v_i,v_{-i})) \cdot(1-p(v_i,v_{-i}))= v_i(f(v'_i,v_{-i})) \cdot(1-p(v'_i,v_{-i}))=0$.

\end{proof}

\subsection{Characterizations: An Analog of Roberts Theorem}\label{subsec-analog}

We now prove an analog of Roberts theorem: we show that if the valuations are unrestricted the only set of social choice functions that are implementable in the percentage fee model are those that maximize weighted versions of the Nash Social Welfare.

Instead of proving this result directly, we prove a meta theorem that provides a one to one and onto correspondence between mechanisms in the percentage fee model and mechanisms in the traditional model. This correspondence will allow us to prove the characterization. However, it does not preserve the approximation ratio.

Thus, in Section \ref{sec-approximations} we devise computationally efficient approximation mechanisms in the percentage fee model.

\begin{definition}
Let $v$ be a positive valuation. Define ${\tt log}v$ to be the valuation such that for every bundle $S$, ${\tt log}v(S)=\log v(S)$. Given a class of positive valuations $\cV$, let ${\tt log} \cV$ denote the class of valuations
$$ {\tt log} \cV = \{ {\tt log} v: v \in \cV \}$$
\end{definition}

\begin{theorem}\label{thm:meta}
Let $\cM$ be an $n$-player mechanism in the percentage fee model that is composed of an allocation function $f:\mathcal V^n\rightarrow \mathcal A$ and a payment function $p$, where $\mathcal V$ is a class of positive valuations. Let $\cM'$ be an $n$-player mechanism in the traditional model with an allocation function $f':{\tt log}\mathcal V^n\rightarrow \mathcal A$ and payment function $p'$. Suppose that for every instance $(v_1,\ldots, v_n)$ and player $i$, $f(v_1,\ldots, v_n)=f'({\tt log}v_1,\ldots,{\tt log} v_n)$ and $p'_i({\tt log}v_1,\ldots, {\tt log}v_n)=\log \frac 1 {(1-p_i(v_1,\ldots, v_n))}$. 

Then, $\cM$ is incentive compatible in the percentage fee model if and only if $\cM'$ is incentive compatible in the traditional model.
\end{theorem}

Note that $\cM$ and $\cM'$ are both well defined since $p_i(v_1,\ldots, v_n)\neq 1$.

\begin{proof}
For simplicity of presentation we prove incentive compatibility for player $1$ only. The proof for the other players is identical. Fix the valuations of all players except player $i$. It suffices to prove that if there exists a profitable deviation in one model than there exists a profitable deviation in the other model as well. Below we use the standard notation $(u,v_{-1})=(u,v_2,\ldots, v_n)$. We will also use the notation  $(u,{\tt log}v_{-1})=(u,{\tt log}v_2,\ldots, {\tt log}v_n)$.
\begin{align*}
\begin{split}
v_1(f(v_1,v_{-1}))\cdot (1- p_i(v_1,v_{-1})) &\geq v_1(f(v'_1,v_{-1}))\cdot  (1-p_i(v'_1,v_{-1}))\\
&\iff\\
\log v_1(f(v_1,v_{-1}))+\log (1- p_i(v_1,v_{-1})) &\geq \log v_1(f(v'_1,v_{-1}))+\log  (1-p_i(v'_1,v_{-1}))\\
&\iff\\
\log v_1(f(v_1,v_{-1}))-\log (\frac 1 {1- p_i(v_1,v_{-1})}) &\geq 
\log v_1(f(v'_1,v_{-1}))-\log  (\frac 1 {1-p_i(v'_1,v_{-1})})\\
&\iff \\
{\tt log} v_1(f'({\tt log}v_1,{\tt log}v_{-1}))-p'_i({\tt log}v_1,{\tt log}v_{-1}) &\geq
 {\tt log} v_1(f'({\tt log}v'_1,{\tt log}v_{-1}))-p'_i({\tt log}v'_1,{\tt log}v_{-1})
\end{split}
\end{align*}

\end{proof}

We say that $\cM$ and $\cM'$ from the statement of the lemma are \emph{twin mechanisms}. Although technically simple, the connection is quite powerful and allows us to easily adapt known results in the traditional model to the percentage fee model:

\begin{theorem}[an analog of Roberts' theorem]
Let $\cM$ be an $n$-player incentive compatible mechanism in the percentage fee model for an unrestricted domain of positive valuations $\cV^n$. Let $\mathcal A$ be the set of alternatives and suppose that the size of the image of the allocation function of $\cM$ is at least $3$. Then, there exist constants $\alpha_a$ (for each $a\in \mathcal A$) and $\beta_i$ (for each player $i$) such that the allocation function of $\cM$ in every instance $(v_1,\ldots, v_n)$ selects an alternative in $\mathcal A$ that maximizes  $\alpha_a\cdot \prod_iv_i(a)^{\beta_i}$.
\end{theorem}

\begin{proof}
Let $\cM'$ be the twin mechanism of $\cM$.  Since $\cV^n$ is an unrestricted domain of positive valuations, ${\tt log }\cV$ is an unrestricted domain (with no positivity condition). By Roberts' theorem, $\cM$ must be an affine maximizer, i.e. it maximizes $\alpha_a + \sum_{i=1}^{n} \beta_i w_i(a)$ for some parameters $\alpha_a, \beta_i$. This also defines the allocation function of the mechanism $\cM$ to be as in the statement of the theorem.
\end{proof}

\subsection{Single Parameter Domains}\label{subsec-single-parameter}

We now study single parameter domains in the percentage fee model. We focus on binary single parameter domains where for each player $i$, the set of alternatives $\mathcal A$ is divided to ``winning'' alternatives $\mathcal W_i$ and ``losing'' alternatives $\mathcal L_i$. For every player $i$ and valuation $v_i\in \mathcal V_i$ there is a value $h_{v_i}$ such that $v_i(a)=h_{v_i}$, for all $a\in \mathcal W_i$. There is also some value\footnote{Note that we assume that $l_i>0$ since if $l_i=0$ only trivial functions can be implemented: fixing $v_{-i}$, player $i$ will get an alternative in $\mathcal L_i$ only if for all $u$ it holds that $f_i(u,v_{-i})\in \mathcal L_i$, as if there is some $u$ for which $f_i(u,v_{-i})\in \mathcal W_i$, then $v_i(f(u,\ldots, v_n))(1-p_i(v_1,\ldots,v_n)>v_i(f(v_i,\ldots, v_n))(1-p_i(v_1,\ldots,v_n)$.} $l_i>0$ such that for every valuation $v_i\in \mathcal V_i$ and alternative $a\in \mathcal L_i$, $v_i(a)=l_i$.

Recall that a social choice function $f$ is \emph{monotone} for player $i$ if, for every $v,u$ for which $h_v<h_u$ and $v_{-i}$ it holds that if $f(v,v_{-i})\in \mathcal W_i$ then $f(u,v_{-i})\in \mathcal W_i$. Recall that in the traditional model monotonicity characterizes implementability. Unlike implementability in rich domains where different functions can be implemented in the traditional model and in the percentage fee model, for single parameter domains we get that the set of implementable social choice functions is identical in the percentage fee model and in the traditional model:

\begin{theorem}
Let $f$ be a social choice function when the domains of all players are single parameter. Then, $f$ is implementable in the percentage fee model if and only if $f$ is monotone for each player $i$.
\end{theorem}
\begin{proof}
Let $\cM$ be a mechanism that implements $f$ with a payment function $p$. To show that $f$ is monotone is will be easier to work with a ``normalized'' mechanism in which the payment for winning player is always $0$. We define  $\cM'$ to be a mechanism with a social choice function $f$ and payment function $p'$ that is defined as follows:
$$
p'_i(v_1,\ldots, v_n)=
\begin{cases}
0 & f(v_1,\ldots, v_n)\in \mathcal W_i;\\
1-\frac {1-p_i(v_1,\ldots, v_n)} {1-p_i(w,v_{-i})} & f(v_1,\ldots, v_n)\in \mathcal L_i \text{ and $\exists w$ s.t. $f(w,v_{-i})\in \mathcal W_i$;}\\
p_i(v_1,\ldots, v_n) & \text{otherwise.}
\end{cases}
$$
Note that $\cM'$ may not be formally a mechanism since it might be that $p'_i(\cdot) \not\in [0,1)$. However, it is still incentive compatible: for every $v_i,u_i,v_{-i}$ such that $f(v_i,v_{-i}) \in \mathcal W_i$ and  $f(u_i,v_{-i}) \in \mathcal L_i$, it holds that:
\begin{align*}
	v_i(f(v_i,v_{-i}))\cdot (1-p_i(v_i,v_{-i})) &\geq v_i(f(u_i,v_{-i}))\cdot (1-p_i(u_i,v_{-i})) \\
	&\iff& \\
	v_i(f(v_i,v_{-i})) \geq v_i(f&(u_i,v_{-i}))\cdot \frac {(1-p_i(u_i,v_{-i}))} {(1-p_i(v_i,v_{-i}))} \\
	& \iff \\
	 v_i(f(v_i,v_{-i}))\cdot (1-p'_i(v_i,v_{-i})) &\geq v_i(f(u_i,v_{-i}))\cdot (1-p'_i(u_i,v_{-i}))
\end{align*}

  We can similarly show that for every $v_i,u_i,v_{-i}$ such that $f(v_i,v_{-i}) ,f(u_i,v_{-i}) \in \mathcal L_i$
\begin{align*}
v_i(f(v_i,v_{-i}))\cdot (1-p_i(v_i,v_{-i})) &\geq v_i(f(u_i,v_{-i}))\cdot (1-p_i(u_i,v_{-i})) 
\\ &\iff 
\\ v_i(f(v_i,v_{-i}))\cdot (1-p'_i(v_i,v_{-i})) &\geq v_i(f(u_i,v_{-i}))\cdot (1-p'_i(u_i,v_{-i}))
\end{align*}
As for $v_i,u_i,v_{-i}$ such that $f(v_i,v_{-i}) ,f(u_i,v_{-i}) \in \mathcal W_i$, we get that $v_i(f(v_i,v_{-i}))\cdot (1-p'_i(v_i,v_{-i})) = v_i(f(u_i,v_{-i}))\cdot (1-p'_i(u_i,v_{-i}))$.

We prove that $f$ is monotone for each player $i$ by using the fact that $\cM'$ implements it. Fix some $v_{-i}$ and let $v_i$ be some valuation such that $f(v_i,v_{-i})\in \mathcal W_i$ and $u_i$ be some valuation such that $f(u_i,v_{-i})\in \mathcal L_i$ (if there is no such $v_i$ or no such $u_i$ then the function is trivially monotone with respect to player $i$ and this $v_{-i}$). We have that $v_i(f(v_i,v_{-i}))\cdot (1-p'_i(v_i,v_{-i})) \geq v_i(f(u,v_{-i}))\cdot (1-p'_i(u,v_{-i}))$ if and only if some alternative in $\mathcal W_i$ is selected. Recalling that $p'_i(v_i,v_{-i})=0$ we get that $h_{v_i}\geq l_i\cdot (1-p'_i(u,v_{-i}))$ if and only if some alternative in $\mathcal W_i$ is selected. That is, for each valuation $v'_i$ for which $h_{v'_i} > h_{v_i}$, an alternative from $\mathcal W_i$ is selected, as needed.
\end{proof}

\section{Computationally Efficient Approximation Mechanisms}\label{sec-approximations}

This section studies the power of polynomial-time incentive-compatible algorithms. The primary tool that we use is maximal-in-range mechanisms. In the percentage fee model, these are mechanisms that find an allocation that maximizes the Nash Social Welfare over some predefined set of allocations. The incentive compatibility of maximal-in-range mechanisms with suitable payments follows from Theorem \ref{thm:VCG}.

First, we develop polynomial-time incentive-compatible mechanisms for a constant number of players $n$ with XOS or subadditive valuations (Subsection \ref{subsec-approx-constant-number}). The approximation ratio of the mechanisms is $n$ for XOS valuations and $n\cdot\log m$ for subadditive valuations. Note that achieving an $n$-approximation via an  incentive-compatible mechanism is easy in the traditional model of social welfare: either conduct a second-price auction on the grand bundle or allocate all items randomly. However, these mechanisms do not approximate the NSW well (a second price auction on the grand bundle provides an approximation ratio of $0$, and a random allocation may provide a reasonable approximation with an exponentially small probability). We develop a new method that provides an $n$-approximation by utilizing $n$-wise independent distributions.

When $n$ is not necessarily a constant, we develop two randomized maximal-in-range algorithms. The first one (Subsection \ref{subsec-approx-m23}), provides a ratio of $O(m^{\frac 2 3})$ in polynomial time with $poly(n,m)$ value queries. The second one (Subsection \ref{subsec-approx-m12}) provides a better ratio of $O(m^{\frac12})$ with $poly(n,m)$ value queries. However, although the number of queries of the latter algorithm is polynomial, its running time is exponential.

Finally, in Subsection \ref{subsec-hardness}, we show that maximal-in-range mechanisms cannot do much better in polynomial time: for any constant number of players $n>1$, there is no maximal-in-range algorithm that provides an approximation ratio better than $n$, unless $NP\subseteq P/poly$.

\subsection{Approximations for a Constant Number of XOS / Subadditive Players}\label{subsec-approx-constant-number}

\begin{theorem}
For $m$ items and $n$ players with XOS valuations, there is a deterministic mechanism in the percentage fee model that guarantees an $(1+\eps)n$-approximation to the Nash Social Welfare (for any constant $\eps>0$). If the valuations are subadditive, then the approximation ratio is $O(n \log m)$. The running time of the mechanism is $O(m^n)$, hence polynomial for a constant number of agents $n$.
\end{theorem} 

We remark that a trivial mechanism maximizing over all possible allocations would have a running time of $O(n^m)$. This is exponential even for a constant number of agents $n$. (We may assume that $m \geq n$, otherwise in every allocation there is one player that is not allocated any item, thus the NSW is always $0$.)

\begin{proof}
Our mechanism is a maximal-in-range mechanism with the following range: Given a set of items $M$ and a set of agents $N$, consider an $n$-wise independent distribution over $N^M$, i.e. a distribution $\cD$ such that for any $j_1,j_2,\ldots,j_n \in M$, $\cD$ projected on $N^{\{j_1,\ldots,j_n\}}$ is a uniform product distribution. Such distributions are known with support of size $O(m^n)$, with exact uniformity for $m,n$ powers of a prime $p$, and with some small deviation from uniformity for general $m,n$.\footnote{Assuming $m > n$ are powers of the same prime $p$, and $\FF_m$ is a field with $m$ elements, we consider a polynomial $p(x) = \sum_{k=0}^{n-1} a_k x^k$ where $a_0,a_1,\ldots,a_{n-1}$ are uniformly random in $\FF_m$. Agent $i$ receives item $j$ if and only if $p(j) = i \pmod n$. The random variables $p(j)$ are $n$-wise independent, and uniformly distributed in $\FF_m$. The size of the probability space is $O(m^n)$, corresponding to the choices of $n$ coefficients in $\FF_q$. For general values of $m,n$, we can choose a prime power $q > m, q = O(m)$ and embed our construction in $\FF_q$, with some small non-uniformity in our distribution (due to $q$ not being divisible by $n$).}
The deviation from uniformity is the reason for $(1+\eps)n$ in the statement, otherwise we get a factor of $n$. In the following, we ignore this issue and assume that we have a uniform $n$-wise distribution $\cD$. 

Given valuations $v_1,\ldots,v_n$, the mechanism maximizes $\prod_{i=1}^{n} v_i(S_i)$ over all allocations $(S_1,\ldots,S_n)$ in the support of $\cD$, which takes $O(m^n)$ running time (by simple exhaustive search). The remaining claim is that this mechanism provides an $n$-approximation to the optimal Nash Social Welfare. We prove that
\begin{equation}
\label{eq:n-approx}
\E_{(S_1,\ldots,S_n) \sim \cD}\left[\prod_{i=1}^{n} v_i(S_i) \right] \geq \frac{1}{n^n} \prod_{i=1}^{n} v_i(O_i)
\end{equation}
where $(O_1,\ldots,O_n)$ is an optimal allocation. Hence, the best allocation in the support of $\cD$ provides an $n$-approximation in terms of the Nash social welfare, $\left( \prod_{i=1}^{n} v_i(S_i) \right)^{1/n} \geq \frac{1}{n} \left( \prod_{i=1}^{n} v_i(O_i) \right)^{1/n}$. \footnote{Note a subtle point here: We cannot claim that a randomized mechanism which samples a random allocation from $\cD$ provides an $n$-approximation in expectation, because $\E[\left( \prod_{i=1}^{n} v_i(S_i) \right)^{1/n}]$ could be substantially smaller than $\left( \E[\prod_{i=1}^{n} v_i(S_i)] \right)^{1/n}$.)}

The proof of (\ref{eq:n-approx}) is as follows: consider an optimal allocation $(O_1,\ldots,O_n)$ and for each agent, an additive valuation in their XOS representation that attains the value of $O_i$:  $v_i(O_i) = \sum_{j \in O_i} w_{ij}$. We also have 
$v_i(S) \geq \sum_{j \in S} w_{ij}$ for every bundle $S$, by the definition of XOS valuations. Hence,
\begin{eqnarray*}
\E_{(S_1,\ldots,S_n) \sim \cD}\left[\prod_{i=1}^{n} v_i(S_i) \right] 
& \geq & \E_{(S_1,\ldots,S_n) \sim \cD}\left[\prod_{i=1}^{n} \sum_{j \in S_i} w_{ij} \right]  \\
& \geq & \E_{(S_1,\ldots,S_n) \sim \cD}\left[\prod_{i=1}^{n} \sum_{j \in S_i \cap O_i} w_{ij} \right]  \\
& = & \E_{(S_1,\ldots,S_n) \sim \cD}\left[\sum_{j_1 \in S_1 \cap O_1} \ldots \sum_{j_n \in S_n \cap O_n} \prod_{i=1}^{n} w_{i j_i} \right]  \\
& = & \sum_{j_1 \in O_1} \ldots \sum_{j_n \in O_n} \Pr[j_1 \in S_1, \ldots, j_n \in S_n] \prod_{i=1}^{n} w_{i j_i}.  \\
\end{eqnarray*}
Now by the property of uniform $n$-wise independent distributions, we have $\Pr[j_1 \in S_1, \ldots, j_n \in S_n] = \prod_{i=1}^{n} \Pr[j_i \in S_i] = \frac{1}{n^n}$. So we can write
\begin{eqnarray*}
\E_{(S_1,\ldots,S_n) \sim \cD}\left[\prod_{i=1}^{n} v_i(S_i) \right] 
& \geq & \sum_{j_1 \in O_1} \ldots \sum_{j_n \in O_n} \frac{1}{n^n} \prod_{i=1}^{n} w_{i j_i} \\
& = & \frac{1}{n^n} \prod_{i=1}^{n} \sum_{j \in O_i} w_{ij} = \frac{1}{n^n} \prod_{i=1}^{n} v_i(O_i)
\end{eqnarray*}
which is the desired inequality (\ref{eq:n-approx}).

Finally, for subadditive valuations, we recall that for every subadditive valuation $v$ there is an XOS valuation $v'$ that approximates it within an $O(\log m)$ factor: for every $S$, $v'(S)\leq v(S)\leq v'(S) \cdot \log m$ \cite{D07}. As a corollary, the same analysis gives a factor of $O(n \log m)$ for subadditive valuations.
\end{proof}

\subsection{A Polynomial Time \texorpdfstring{$O(m^{\frac 2 3})$}{O(m\^{}(2/3))}-Approximation Mechanism}\label{subsec-approx-m23}

We now provide a randomized maximal-in-range mechanism that provides an approximation ratio of $O(m^{\frac 2 3})$ for subadditive valuations. The mechanism combines  two maximal-in-range algorithms. The first mechanism finds an allocation that maximizes the NSW among all allocations that allocate one item to each player. This mechanism provides an approximation ratio of $\frac m n$. Note that this mechanism can be implemented by running a bipartite matching algorithm on the graph that contains players on one side and items on the other. The graph has an edge between player $i$ and item $j$ if $v_i(\{j\})>0$ and in this case its weight is $\log v_i(\{j\})$ .

The second mechanism partitions the set of items to $n^2$ bundles $S_1,...,S_n$, by allocating each item to one of the bundles independently at random and finding the allocation that maximizes the NSW among all allocations in which each player gets one such bundle. This allocation can be found by running a bipartite matching algorithm, similarly to before. We show that this mechanism provides an approximation ratio of $O(n^2)$. Intuitively, the idea of the algorithm is as follows. Consider some player $i$ and the bundle $O_i$ that he gets in the optimal solution. For each player $i$ there is at least one bundle $S_i$ that is valuable enough for her, since by subadditivity there is at least one bundle $S_j$ such that $v_i(S_j\cap O_i)\geq \frac{v_i(O_i)}{n^2}$. We say that player $i$ is interested in the bundle $S_j$. A simple application of the birthday paradox shows that since the bundle $S_j$ that player $i$ is interested in is distributed uniformly and independently, the probability that no two players are interested in the same bundle $S_j$ is at most $\frac 1 2$. Thus we can allocate each player $i$ the bundle $S_j$ that she is interested in and get an approximation ratio of $\frac 1 {n^2}$.

\subsubsection*{The Algorithm}

\begin{itemize}
\item Randomly partition the items into $n^2$ bundles, $S_1,\ldots, S_{n^2}$. Each item $j$ is assigned to a set uniformly and independently of the other items.
\item Choose the allocation that maximizes the Nash Social Welfare from the union of the following sets of allocations:
\begin{enumerate}
\item The set of allocations which consists of all allocations in which each player gets one item.
\item The set of allocations which consists of all allocations in which each player gets one of the bundles $S_1,\ldots, S_{n^2}$.
\end{enumerate}
\end{itemize}

\begin{theorem}
There are payments that make the mechanism above incentive compatible in the percentage fee model. The mechanism guarantees an approximation ratio of $\min(\frac m n, n^2)=m^{\frac 2 3}$ with probability at least $\frac 1 2$ whenever the valuations of the players are subadditive. There is an implementation of the mechanism that uses polynomially many value queries and runs in polynomial time.
\end{theorem}

It is easy to implement the mechanism with polynomially many value queries: query each player $i$ for her value $v_i(\{j\})$, for every item $j$, and for her value for the bundles $S_1,\ldots, S_{n^2}$. The total number of queries is at most $n\cdot (m+n^2)$. Note that given these queries we can find the allocation that maximizes the NSW in polynomial time, by running a bipartite matching twice, once for each set of allocations. Incentive compatibility also follows since the mechanism is maximal-in-range (in the percentage fee model). 

As stated above, the mechanism is randomized and succeeds with probability $\frac 1  2$. Thus, we can run the mechanism $t$ times and the success probability increases to $1-(\frac 1 2)^t$. That is, the failure probability is $1/exp(m,n)$ when $t$ is large enough but still polynomial. In general, incentive compatibility is not preserved when repeating an incentive compatible mechanism several times and choosing the best outcome. However, the incentive compatibility of maximal-in-range mechanisms (and hence also of this specific mechanism) is preserved: randomization is used only to define the range of the mechanism. When payments are computed with respect to the union of the ranges produced by the different coin flips, the composed mechanism is incentive compatible.

It remains to prove that the mechanism provides the required approximation ratio. This will be proved by the following two lemmas by noting that $\min(\frac m n, n^2)\leq m^{\frac 2 3}$ for all possible values of $m,n$.

\begin{lemma}\label{algs-approx-mn}
With probability $1$, the mechanism returns a solution of NSW at least $\frac{n}{m} OPT$.
\end{lemma}
\begin{proof}
We will show that this approximation ratio is guaranteed even if we consider only the first set of allocations, where each player gets at most one item. 

Let $(O_1,\ldots, O_n)$ denote an allocation that maximizes the NSW. For each $O_i$, let $o_i$ denote an item in $O_i$ with the largest value as a singleton, i.e., $o_i\in \arg\max_{j\in O_i}v_i(\{j\})$. Let $(\{a_1\},\ldots, \{a_n\})$ denote the allocation that maximizes the NSW in the first set of allocations. We have that:
$$
\prod_{i=1}^{n} v_i(\{a_i\}) \geq \prod_{i=1}^{n}  v_i(\{o_i\}) \geq \prod_{i=1}^{n}  \frac {1} {|O_i|} \cdot v_i(O_i) \geq \left(\frac n m \right)^n \prod_{i=1}^{n}   v_i(O_i)
$$
where the second inequality is due to subadditivity and the last inequality follows from $\Sigma_i|O_i|\leq m$ and the AM-GM inequality.
\end{proof}

\begin{lemma}
With probability at least $\frac12$, the mechanism returns a solution of NSW at least $\frac{1}{n^2} OPT$.
\end{lemma}
\begin{proof}
We will show that this approximation ratio is guaranteed with the specified probability even if we consider only the second set of allocations.

Let $(O_1,\ldots, O_n)$ denote an allocation that maximizes the NSW. By subadditivity, we are guaranteed that for each player $i$ there exists at least one index $j_i \in [n^2]$ such that $v_i(O_i \cap S_{j_i})\geq \frac {v_i(O_i)} {n^2}$ (if there are several such $j_i$'s, choose one uniformly at random). Note that since each item is assigned to one of $S_1,\ldots, S_{n^2}$ independently, the $j_i$'s are independent and uniform in $\{1,2,\ldots, n^2\}$. The rest of the proof follows from the following claim, which is essentially the (flip side of the) birthday paradox:
\begin{claim*}
With probability at least $\frac 1 {2}$, there are no two players $i,i'$ such that $j_i=j_{i'}$.
\end{claim*}
\begin{subproof}
Fix two players, $i,i'$. Since $j_i, j_{i'}$ are chosen independently and uniformly from a range of $n^2$ values, the probability that $j_i=j_{i'}$ is $\frac 1 {n^2}$. The number of pairs is $\binom{n}{2}$, so by the union bound, the probability that there exists a pair of players $i \neq i'$ such that $j_i=j_{i'}$ is at most $\frac {\binom{n}{2}} {n^2} < \frac 1 2$.
\end{subproof}

Let $ALG$ denote the NSW of the best allocation in the range. This is at least as good as the allocation that gives each player $i$ the bundle $S_{j_i}$, which provides an approximation ratio of $n^2$:
$$
ALG \geq \left( \prod_{i=1}^{n}  v_i(S_{j_i}) \right)^{1/n} \geq \left( \prod_{i=1}^{n}  \frac {v_i(O_i)}{n^2}  \right)^{1/n}  \geq \frac{1}{n^2} OPT.
$$
\end{proof}

\subsection{An \texorpdfstring{$\tilde O(m^{\frac12})$}{\~O(m\^{}(1/2))}-Approximation Mechanism using \texorpdfstring{$O(n+m)$}{O(n+m)} Value Queries}\label{subsec-approx-m12}

We now show how to improve the approximation ratio of the algorithm provided in Subsection \ref{subsec-approx-m23}. However, while the number of queries the algorithm of this subsection makes is still polynomial, the running time is not polynomial. 

We again combine two maximal-in-range algorithms. The first one finds an optimal matching and provides an approximation ratio of $\frac m n$ exactly as before. In the second mechanism, we partition the items uniformly and independently into $t=2n$ bundles, $S_1,\ldots, S_t$. Our range consists of allocations where each player gets either one of the bundles $S_j$ or a singleton item. We show that the approximation factor of this mechanism is $\tilde{O}(n)$ w.h.p.

Consider an optimal allocation $(O_1,\ldots,O_n)$. We show that if player $i$ does not have an item $j\in O_i$ that contributes  $\tilde{\Omega}(\frac{1}{t} v_i(O_i))$, then $\E[v_i(S_i \cap O_i)] = \tilde{\Omega}(\frac{1}{t} v_i(O_i))$ with high probability. So if a typical bundle $S_j$ is not valuable for player $i$ with high probability, then there is a ``significant item'' which is valuable for player $i$ (if there are several such item, arbitrarily choose one item the ``significant item'' of player $i$). Thus, the following allocation gives a good approximation with high probability: if player $i$ has a significant item, allocate this item to the player. For the remaining players, allocate them an arbitrary bundle that does not contain items that are significant for some player. Note that since there are at most $n$ players and thus $n$ items that are significant for some player, there are at most $n$ bundles that contain items that are significant for some player. Hence, at most $n$ out of the $2n$ bundles can be invalidated in this way, and so at least $n$ bundles still remain available to be allocated to players that do not have significant items.

We note that the first type of allocation (one item for each player) is actually a special case of the second type, so we state only the second type in our algorithm.

\subsubsection*{The Algorithm}

\begin{itemize}
\item Randomly partition the items into $t=2n$ bundles, $S_1,\ldots, S_{t}$. Each item $j$ is assigned to $S_i$ for a random $i \in [t]$ uniformly and independently of the other items.
\item Choose an allocation that maximizes the Nash Social Welfare over the following set of allocations:
All allocations in which each player $i$ either gets some bundle $S_j$, or a single item outside of the bundles allocated to other players.
\end{itemize}

\begin{theorem}
The mechanism above with suitable payments is incentive-compatible in the percentage fee model, and guarantees an approximation ratio of $\min(\frac m n, O(n\cdot \log^2m)) = O({m}^{1/2} \cdot \log m)$ with probability at least $1-1/m$ whenever the valuations of the players are subadditive. There is an implementation of the mechanism that uses $O(m+n)$ value queries, and the running time is $2^{O(n)} \poly(m)$.
\end{theorem}

It is easy to implement the mechanism with $O(m+n)$ value queries: Given the random partition $(S_1,\ldots,S_n)$, query each player $i$ for her value $v_i(\{j\})$ for every item $j$, and her value $v(S_i)$ for every bundle $S_i$, $1\leq i\leq t$. The total number of such queries is $m+2n$. To find the best allocation in the range, we enumerate over all subsets $A$ of players who should get a full bundle, and over all subsets $B$ of $|A|$ bundles to be allocated as a full bundle: these are $2^{O(n)}$ configurations to consider. For each configuration, we find the best assignment $B$ to $A$ by solving a max-weight matching problem with weights $\log v_i(S_j)$ (if $v_i(S_j)>0$), and also the best assignment of singletons outside of $\bigcup_{i \in B} S_i$ to the players outside of $A$, by solving another matching problem with weights $\log v_i(\{j\})$ (if $v_i(\{j\})>0$). This takes $\poly(m,n)$ time for each configuration. So the total running time is $2^{O(n)} \poly(m)$.

It remains to prove that the mechanism provides the required approximation ratio. We will use the following lemma:

\begin{lemma}\label{lemma:subadditive-lower-tail}
For some constant $c>0$, the following holds. Let $v:2^S \to \R_+$ be subadditive and let $S'$ be a random subset of $S$ that is obtained by including each item of $S$ independently with probability $\frac 1 t$, for some $t>1$. Suppose that for every  $j \in S$, $v(\{j\}) \leq \frac {v(S)} {c\cdot t\cdot \log^2 m}$. Then, 
$$\Pr\left[v(S') \geq \frac {v(S)} {c\cdot t\cdot \log m} \right] > 1 - \frac{1}{m^2}.$$
\end{lemma}

\begin{proof}
We use the following claim from \cite{D07}, which is essentially a result on approximation of subadditive functions by XOS ones.

\begin{claim*}[\cite{D07}]\label{claim-xos-approx}
There exists a constant $c'>0$ such that for every subadditive function $v:2^S \to \R_+$, there exist prices $(p_j: j \in S)$ such that
\begin{enumerate}
\item $\sum_{j \in S} p_j \geq \frac{v(S)}{c'\cdot \log |S|}$,
\item $\forall S' \subseteq S$,  $\sum_{j \in S'} p_j \leq v(S')$.
\end{enumerate}
\end{claim*}

We choose the $c$ in the lemma as $c = 6c'$. Given $S$, consider the prices $p_j$ given by the claim. 
We will give a lower bound on the expected value of $v(S')$, where each element of $S$ appears independently with probability $1/t$. We will use the following Chernoff bound:
$$
\Pr[X<(1-\eps)\mu] < e^{-\eps^2 \mu / 2}
$$
where $X = \sum_{j \in S} a_j X_j$, $0 \leq a_j \leq 1$, $\{X_j: j \in S\}$ are independent random values in $\{0,1\}$, and $\mu=\E[X]$. 

In our setting, $X_j$ is the indicator variable of the event $j \in S'$, and we set $a_j =  \frac{c \cdot t \cdot \log^2 m} {v(S)}  p_j$. Thus, by assumption we have
$a_j \leq  \frac{c \cdot t \cdot \log^2 m} {v(S)}  v(\{j\}) \leq 1$, $\mu = \E[X] = \frac{1}{t} \sum_{j \in S} a_j = \frac{6c' \cdot \log^2 m} {v(S)}  \sum_{j \in S} p_j \geq 6 \log m$ and hence by the Chernoff bound with $\epsilon=5/6$,
$$
\Pr[X < \log m] \leq \Pr[X<\mu/6] < e^{-(5/6)^2 \mu / 2} < e^{-\mu/3} \leq \frac{1}{m^2}.
$$
Consequently, with probability more than $1-1/m^2$, $X \geq \log m$, and
$$
v(S') \geq \sum_{j \in S'} p_j = \frac{v(S)}{c \cdot t \cdot \log^2 m} \sum_{j \in S'} a_j
 = \frac{v(S)}{c \cdot t \cdot \log^2 m} X \geq \frac{v(S)}{c \cdot t \cdot \log m}.
$$
\end{proof}

The approximation ratio can now be proved by combining the following two lemmas.

\begin{lemma}
With probability $1$, the mechanism outputs an allocation of value at least $\frac{n}{m} OPT$.
\end{lemma}

The proof of this Lemma is identical to the proof of Lemma \ref{algs-approx-mn}, using the fact that allocations of 1 item to each player are included in our range:
We select a set of players $A$ who obtain a full bundle $S_j$, and the other agents obtain a singleton item. As a special case, we consider $A = \emptyset$, in which case the allocation is simply a matching.

\begin{lemma}
With probability $1-1/m$, the mechanism outputs an allocation of value at least  $\frac{1}{2c n \log^2 m} OPT$.
\end{lemma}

\begin{proof}
Let $(O_1,\ldots, O_n)$ denote an allocation that maximizes the NSW. We call player $i$ \emph{focused} if there exists some item $s_i \in O_i$ such that $v_i(\{s_i\})\geq \frac {v_i(O_i)} {c\cdot t \cdot \log^2 m}$. In this case, we call $s_i$ the \emph{significant} item of player $i$ (if there are several such items, we choose one arbitrarily). If the significant item $s_i$ of a focused player $i$ is assigned to a set $S_j$, we say that player $i$ is \emph{interested in} the set $S_j$.

Let's call a bundle available if it doesn't contain any significant item.  Observe that at most $n$ of the bundles $S_1,\ldots,S_t$ can contain a significant item of some player, and hence at least $n$ bundles are available (recall that $t=2n$). Let us allocate some available bundle $S_{\sigma(i)}$ to each player $i$ who is not focused; we can do this uniformly at random, given the set of available bundles. Since items are assigned to bundles $S_j$ independently, the choice of $\sigma(i)$ is independent of how $O_i$ is partitioned among the bundles. Hence, from the point of view of an unfocused player $i$ (assuming that she cares only about the items in $O_i$), the choice of $\sigma(i)$ is uniformly random, and the set $O_i \cap S_{\pi(i)}$ can be viewed as sampling elements of $O_i$ with probability $1/t$. Hence, we can apply Lemma~\ref{lemma:subadditive-lower-tail}.

\begin{claim*}\label{algs-corollary-non-focused-value}
Assuming $OPT>0$, with probability at least $1-\frac 1 {m}$, for each player $i$ that is not focused it holds that $v_i(S_{\sigma(i)}) \geq \frac {v_i(O_i)} {c\cdot t\cdot \log m}$.   
\end{claim*}

\begin{subproof}
For each player $i$ that is not focused, the value of every singleton is bounded by $v_i(\{j\}) < \frac{v_i(O_i)}{c \cdot t \cdot \log^2m}$.
As discussed above, $O_i \cap S_{\sigma(i)}$ is a random subset of $O_i$ where each item appears independently with probability $1/t$.
By Lemma \ref{lemma:subadditive-lower-tail}, 
 $\Pr[v_i(O_i \cap S_{\sigma(i)}) < \frac {v_i(O_i)} {c \cdot t \cdot\log m}]<\frac 1 {m^2}$. By the union bound, since there are at most $n$ players that are not focused, the statement of the lemma holds with probability at least $1-\frac n {m^2}\geq 1 - \frac 1 m$
(where in the last inequality we use $n\leq m$ otherwise in any allocation at least one player gets no items and $OPT =0$).
\end{subproof}

We can now finish the proof of the lemma by considering the following allocation: Each player $i$ that is focused receives its significant item $s_i$. Each player $i$ that is not focused receives the bundle $S_{\sigma(i)}$. By construction, these are disjoint sets and hence this is a valid allocation. With probability at least $1-1/m$, every player $i$ receives value at least $\frac{v_i(O_i)}{c \cdot t \cdot \log^2 m}$, either from their bundle or their significant item.
Hence, the allocation provides an approximation factor of $\frac{1} {2c \cdot n \log^2 m}$.
\end{proof}

To conclude, we note that the mechanism achieves simultaneously an approximation factor of $\frac{n}{m}$ and $\frac{1}{2c n \log^2 m}$. 
Hence, we have a factor at least $\max \{\frac n m, \frac{1}{n \log^2 m} \} \geq \frac{1}{m^{1/2} \log m}$ for all possible values of $m,n$ (the worst case being $n = \frac{m^{1/2}}{\log m}$).

\subsection{An Impossibility Result for MIR Mechanisms with Additive Valuations}\label{subsec-hardness}

We now prove that the Nash Social Welfare cannot be maximized to within a factor better than $1/n$ by a maximal-in-range mechanism in polynomial time. The proof is composed of two steps. In the first step we show that if the range of some mechanism contains many allocations then there is a relatively large subset of the items $S$ and a set of players $T$ such that projecting the set of all allocations of the mechanism on the subset $S$, we get all possible allocations of the items in $S$ to the players in $T$. The proof relies on similar lemmas that were obtained to prove the limits of polynomial time maximal-in-range mechanisms for maximizing the social welfare. We then use these sets $S$, $T$ to show that the maximal-in-range mechanism must solve exactly the problem of maximizing the NSW for two players with additive valuations, which is NP-hard.

\begin{theorem}
Fix a constant $\eps>0$ and $n \geq 2$. There is no polynomial-time maximal-in-range mechanism for $n$ players with additive valuations that provides a $(\frac{1}{n}+\eps)$-approximation to the Nash social welfare, unless $NP\subseteq P/poly$.
\end{theorem}

\begin{definition}
Let $\cR$ be a set of allocations. We say that $\cR$ contains an $(S,T)$-shattering if the range $\cR$ restricted to the set of items $S$, contains $T^S$, i.e.~ all possible allocations of items in $S$ to players in $T$.
\end{definition}

The next lemma follows from a similar lemma in \cite{BDFK10}.

\begin{lemma}
\label{lemma:shattering}
Suppose that $|M|=m$, $|N|=n$, $\eps \in (0,\frac14)$, $m \geq \frac{2}{\eps^2} n \log (2n)$. Then there exists $\delta = \delta(n,\epsilon)>0$ such that if $\cM$ is a maximal-in-range mechanism with range $\cR \subseteq N^M$ that provides a $(\frac{1}{n} + \eps)$-approximation in terms of Nash social welfare for additive valuations, then there exists a subset of items $S\subseteq M$, $|S| \geq \delta m$, and a subset of the players $T \subseteq N$, $|T|\geq 2$, such that $\cR$ contains an $(S,T)$-shattering.
\end{lemma}

\begin{proof}
Consider a uniformly random function $f:M \to N$. We interpret $f$ as an instance of Nash Social Welfare, where $f^{-1}(i)$ is the set of items in which player $i$ is interested, and her (additive) valuation is $v_i(S) = |S \cap f^{-1}(i)|$. 
We assume that $\cM$ provides a $c$-approximation, hence there must be an allocation in the range $(S_1,\ldots,S_n) \in \cR$ such that
\begin{equation}
\label{eq:c-approx}
 \left( \prod_{i=1}^{n} v_i(S_i) \right)^{1/n} = \left( \prod_{i=1}^{n} |S_i \cap f^{-1}(i)| \right)^{1/n}
 \geq \left( \frac{1}{n} + \eps \right) \cdot OPT = \left( \frac{1}{n} + \eps \right) \cdot \left( \prod_{i=1}^{n} |f^{-1}(i)| \right)^{1/n}
\end{equation}
This holds for every instance $f$, but we are particularly interested in those instances where the sets $f^{-1}(i)$ are ``$\eps$-balanced'': $|f^{-1}(i)| \geq (1-\eps) \frac{m}{n}$ for every $i$. By Chernoff bounds, this happens for a uniformly random $f$ with probability at least $1/2$:
For each $i$, $|f^{-1}(i)|$ is a summation of independent 0/1 random variables with expectation $\mu = m/n$, and by the Chernoff bound, $\Pr[|f^{-1}(i)| < (1-\eps) \mu] < e^{-\eps^2 \mu / 2} \leq e^{-\log (2n)} = 1/(2n)$ by the assumptions of the lemma.

Assuming that $f$ is $\eps$-balanced, (\ref{eq:c-approx}) implies the following: By the AM-GM inequality,
\begin{align}
\label{eq:sum-bound}
\frac{1}{n} \sum_{i=1}^{n} |S_i \cap f^{-1}(i)| &\geq  \left( \prod_{i=1}^{n} |S_i \cap f^{-1}(i)| \right)^{1/n}
 \geq \left( \frac{1}{n} + \eps \right) \cdot \left( \prod_{i=1}^{n} |f^{-1}(i)| \right)^{1/n}\nonumber \\
  &\geq   \left( \frac{1}{n} + \eps \right) (1-\eps) \frac{m}{n} > \left( \frac{1}{n} + \frac{\eps}{2} \right) \frac{m}{n}.
\end{align}

Let us define a function $g:M \to N \cup \{*\}$ encoding $(S_1,\ldots,S_n)$: $g(j) = i$ if $j \in S_i$ and $g(j) = *$ if $j$ is not contained $\bigcup_{i=1}^{n} S_i$. Equation (\ref{eq:sum-bound}) can be interpreted as saying that $f$ and $g$ agree on at least $\left( \frac{1}{n} + \frac{\eps}{2} \right) m$ coordinates, i.e. the Hamming distance between $f$ and $g$ is at most $(1 - \frac{1}{n} - \frac{\eps}{2}) m$. Now we refer to Lemma 4.4 in \cite{BDFK10}, with $U=M$, $V=N$, $\gamma = 1/2$, $q=2$:  The lemma concludes that for some $\delta(n,\epsilon)>0$ there is a set of items $S \subseteq M$ and a set of players $T \subseteq N$ such that $|S| \geq \delta |M|$, $|T| \geq 2$, and $\cR$ contains an $(S,T)$-shattering.
\end{proof}

\begin{lemma}
Let $\cM$ be a maximal-in-range mechanism for $n$ players, for some constant $n \geq 2$. Suppose for some fixed $\delta>0$ and every $m$, the range of $\cM$ contains an $(S,T)$ shattering, for $|S| \geq \delta  m$ and $|T| \geq 2$. Then, $\cM$ does not run in polynomial time, unless $NP\subseteq P/poly$.
\end{lemma}

\begin{proof}
We use $\cM$ to exactly solve the problem of maximizing the NSW with a set of $2$ additive players and $m'$ items. Recall that this problem is NP-hard (by a simple reduction from the NP-complete Subset Sum problem). Let $(v'_1,v'_2)$ be such an instance on $m'$ items.

Given $m'$, we consider instances of NSW with $n$ players and $m = \lceil m' / \delta \rceil$ items. Our assumed MIR mechanism with these parameters contains an $(S,T)$-shattering where $|S| \geq \delta m \geq m'$ and $|T| \geq 2$. We can actually assume that $|T|=2$, which is implied by any shattering with a larger $T$. Let us also assume for convenience that $S = [m'] = \{1,2,\ldots, m'\}$ and $T = \{1,2\}$. 

Given the valuations $v'_1,v'_2$ on $[m']$, we define additive valuations $v_1,v_2$ on $[m]$ as follows: $v_i(j)=v'_i(j)$, for $1\leq j\leq m'$, and $v_i(j)=0$ otherwise. To extend this to an instance with $n$ players, we need to define the valuations of players $3,\ldots, n$. For each such player $i$, we will have a set of $(1-\eps) m$ permissible valuations $v^j_i$: For every $m' < j \leq m$, we set $v^j_i(\{j\})=1$ and $v^j_i(\{j'\})=0$, if $j'\neq j$.

Note that the set of permissible valuations for each extra player is of size $(1-\eps) m$ and thus there are $((1-\eps)m)^{n-2}$ instances that may be obtained by each extra player having one valuation from its permissible set. 
We run $\cM$ on all such instances and choose an allocation that maximizes the NSW. Thus, the total running time of the reduction is $O(m^{n-2})$ times the running time of the mechanism $\cM$, which is polynomial assuming that $n$ is constant and the running time of $\cM$ is polynomial as well. 
The correctness of the reduction follows from the following claim.

\begin{claim*}
In all iterations of the reduction, the optimal NSW of the constructed instance is at most the optimal NSW of the original instance. Further, there exists an iteration of the reduction in which the optimal NSW equals the optimal NSW in the original instance. 
\end{claim*}
\begin{subproof}
Consider some iteration of the reduction, where the choice of the valuation of each player $i$, $n\geq i>2$, is $v_i^j$. Consider some allocations of the items $(S_1,\ldots, S_n)$. Note that $\prod_{i=3}^{n} v_i^j(S_i)\in \{0,1\}$, by construction. Also note that $\prod_{i=1}^{2} v_i(S_i)=\prod_{i=1}^{2} v_i(S_i\cap [m'])=\prod_{i=1}^{2} v'_i(S_i\cap [m'])$. We thus have that $\prod_{i=1}^{n} v_i(S_i) \leq \prod_{i=1}^{2} v'_i(S_i\cap [m'])$.  Note that $(S_1\cap [m'], \ldots, S_n \cap [m'])$ is an allocation of the items in $[m']$. Hence, the first part of the claim follows.

We now prove the second part of the claim. Let $(O_1,O_2)$ be an optimal allocation of items in $[m']$ to players in $T=\{1,2\}$. Note that there since the range has an $(S,T)$-shattering with $S = [m'], T=\{1,2\}$, there is an allocation $(S_1,\ldots, S_n)$ such that $\prod_{i=1}^{2} v_i(S_i \cap [m']) = \prod_{i=1}^{2}{v'_i}(O_i)$. We can assume that for each $i > 2$, $S_i \neq \emptyset$, otherwise the NSW of $(S_1,\ldots, S_n)$ is always $0$ and thus we can assume that this allocation is not in the range of the algorithm in the first place. Consider the iteration where the valuation of each player $i>2$ is $v_i^j$, for some $j \in S_i$. We have that $\prod_{i=3}^{n} v_i^j(S_i)=1$. In total we get that $\prod_{i=1}^{n} v_i(S_i) = \prod_{i=1}^{2} v'_i(O_i)$, as needed.
\end{subproof}
\end{proof}

To summarize, the maximum NSW over all allocations in the range would be one whose value is exactly $v'_1(O_1)\cdot  v'_2(O_2)$ and we could also find $O_1, O_2$ (or another allocation of equal NSW value) by restricting the output of our mechanism to the first two players. Hence, we would be able to solve an NP-hard problem for every given input size (possibly non-uniformly, hence the conclusion is $NP \subset P/poly$).

\section{Conclusion and Future Directions}

In this work we design incentive compatible mechanisms that maximize the Nash Social Welfare by considering a novel percentage fee model. Our work leaves a number of open questions. At the most immediate level, can we obtain an approximation ratio of $m^{\frac 1 2}$ not just with a polynomial number of value queries but also in polynomial time? In addition, our hardness result applies only to maximal-in-range mechanisms; is it possible to prove that obtaining an approximation ratio better than $m^\frac 1 2$ for {\em any incentive-compatible mechanism} in our model is computationally hard, or requires an exponential number of queries? Such impossibilities are known in the traditional model \cite{D11,DV11,DV12} but we do not know how to obtain analogous results for approximating the NSW in the percentage fee model.

In addition, all of our bounds use simple value queries. We do not know whether more complicated queries, e.g., demand queries, can help obtain better approximation ratios.

Finally, in this paper we have demonstrated how different payment schemes enable the implementation of useful social choice functions. Is it always possible to characterize the set of implementable social choice functions as a function of the payment method? Specifically, what can be implemented if the designer is allowed to offer, for each alternative, \emph{either} a fixed fee or a percentage fee? It will also be very interesting to understand whether there are other natural payment schemes that enable the incentive-compatible implementation of different fairness notions.

\printbibliography    
\end{document}